\documentclass[aps,pra,reprint,twocolumn,superscriptaddress]{revtex4-1}
\usepackage{amsfonts}
\usepackage{amssymb}
\usepackage{graphicx}
\usepackage{amsmath}
\usepackage{mathrsfs}
\usepackage{amsthm}
\usepackage{bbm}
\usepackage{hyperref}

\newtheorem{theorem}{Theorem}

\newtheorem{lemma}[theorem]{Lemma}
\newtheorem{proposition}[theorem]{Proposition}

\begin{document}

\title{On the constrained classical capacity of infinite-dimensional
covariant quantum channels}
\author{A. S. Holevo}
\affiliation{Steklov Mathematical Institute, 119991 Moscow, Russia.}

\begin{abstract}
The additivity of the minimal output entropy and that of the $\chi$-capacity
are known to be equivalent for finite-dimensional irreducibly covariant quantum
channels. In this paper we formulate a list of conditions allowing to
establish similar equivalence for infinite-dimensional covariant channels
with constrained input. This is then applied to Bosonic Gaussian channels
with quadratic input constraint to extend the classical capacity results of
the recent paper \cite{ghg} to the case where the complex structures
associated with the channel and with the constraint operator need not
commute. In particular, this implies a multimode generalization of the
''threshold condition'', obtained for single mode in \cite{schaefer}, and
the proof of the fact that under this condition the classical ''Gaussian
capacity'' resulting from optimization over Gaussian inputs is equal to the
full classical capacity.

We also investigate implications of the gauge-covariance condition for single- and multimode
Bosonic Gaussian channels.
\end{abstract}
\maketitle

\section{Introduction: finite dimensions}

In classical information theory, the capacity is a unique characteristic of a communication channel.
On the contrary, quantum channel is characterized by a whole variety of entropic quantities.
A powerful tool in investigation of the \textit{classical capacity} of a quantum channel is the \textit{minimal output entropy}.
For example, a theorem in the seminal paper \cite{Shb} allowed to show that failure of the global additivity of the latter implies
a similar phenomenon for the former.
The present paper is devoted to further investigation of the relation between these important quantities on the level of
individual channels.

For the background of this section we refer to \cite{HG}, \cite{h}. Let $
\Phi $ be a quantum channel in $d$-dimensional Hilbert space $\mathcal{H}$
i.e. a linear completely positive trace-preserving map of the algebra of all
linear operators in $\mathcal{H}$. A quantum analog of the Shannon capacity
is the $\chi$-\textit{capacity} of the channel $\Phi $:
\begin{eqnarray}\label{chicap}
&&C_{\chi }(\Phi )=\\
&&\max_{\pi}\left\{ S\left( \Phi \left[ \sum_{x}\pi (x)\rho
(x)\right] \right) -\sum_{x}\pi (x)S(\Phi \lbrack \rho (x)])\right\} \notag,
\end{eqnarray}
where $S(\rho )=-\mathrm{Tr}\rho \log \rho $ is the von Neumann entropy, and
the maximum is over \textit{state ensembles} i.e. finite probability
distributions $\pi $ ascribing probabilities $\pi (x)$ to density operators $%
\rho (x).$

The \textit{classical capacity} of the quantum channel $\Phi $, defined as
the maximal transmission rate per use of the channel, with coding and
decoding chosen for increasing number $n$ of independent uses of the channel
\begin{equation*}
\Phi ^{\otimes n}=\underset{n}{\underbrace{\Phi \otimes \dots \otimes \Phi }}
\end{equation*}%
such that the error probability goes to zero as $n\rightarrow \infty $, is
given by HSW theorem:
\begin{equation}\label{hsw}
C(\Phi )=\mathrm{lim}_{n\rightarrow \infty }(1/n)C_{\chi }(\Phi ^{\otimes
n}).
\end{equation}
In the case where the $\chi $-capacity is additive,
\begin{equation}
C_{\chi }(\Phi ^{\otimes n})=nC_{\chi }(\Phi ),  \label{addi}
\end{equation}%
one has $C(\Phi )=C_{\chi }(\Phi )$. Unlike the classical case, this
property does not hold in general, due to the possibility of \textit{
entangled} encodings at the input of the channel $\Phi ^{\otimes n}$ \cite{hast}.

An obvious upper estimate for $C_{\chi }(\Phi )$ is
\begin{equation}
C_{\chi }(\Phi )\leq \max_{\rho}S\left( \Phi \left[ \rho \right] \right) -%
\check{S}(\Phi ),  \label{ine1}
\end{equation}
where the minimal output entropy of the quantum channel $\Phi $ is defined
as
\begin{equation*}
\check{S}(\Phi )=\min_{\rho }S(\Phi (\rho )).
\end{equation*}

The first term in the right-hand side of (\ref{maxs}) is additive for all channels:
\begin{equation}
\max_{\rho ^{(n)}}S\left( \Phi ^{\otimes n}\left[ \rho ^{(n)}\right] \right)
=n\max_{\rho }S\left( \Phi \left[ \rho \right] \right).  \label{maxs}
\end{equation}
This is a simple corollary of (sub)additivity of the
von Neumann entropy with respect to tensor products (see also lemma \ref{L6}
below).

For some channels (\ref{ine1}) may become equality, allowing to reduce the
additivity (\ref{addi}) of $C_{\chi }$ to the additivity of the minimal output entropy
\begin{equation}
\check{S}\left( \Phi ^{\otimes n}\right) =n\check{S}\left( \Phi \right) .
\label{maddi}
\end{equation}
This is the case for irreducibly covariant channels. Channel $\Phi $ is
\textit{covariant} if there is a continuous (projective) unitary
representation $g\rightarrow V_{g}$ of a symmetry group $G$ in $\mathcal{H}$
such that
\begin{equation}
\Phi \left[ V_{g}\rho V_{g}^{\ast }\right] =U_{g}\Phi \left[ \rho \right]
U_{g}^{\ast },  \label{cova}
\end{equation}%
where $U_{g}$ are unitary operators, and \textit{irreducibly covariant} if
the representation $g\rightarrow V_{g}$ is irreducible. In this case,
assuming compactness of $G$, one has for arbitrary density operator $\rho$
\begin{equation*}
I/d = \int_{G}V_{g}{\rho }V_{g}^{\ast }\,\pi ^{0}(dg) ,
\end{equation*}
where $I$ is the unit operator in $\mathcal{H}$, $\pi ^{0}(dg)$ is the
invariant probability measure on $G$ (this follows from the orthogonality
relations for irreducible representation). Then one can show, see e.g. \cite
{hcov}, that
\begin{equation}
\max_{\rho }S\left( \Phi \left[ \rho \right] \right) =S(\Phi \left[ I/d%
\right] ),  \label{unit}
\end{equation}%
and
\begin{equation}
C_{\chi }(\Phi )=\max_{\rho }S\left( \Phi \left[ \rho \right] \right) -%
\check{S}(\Phi )=S\left( \Phi \lbrack I/d]\right) -\check{S}\left( \Phi
\right) ,  \label{equ}
\end{equation}
making (\ref{ine1}) the equality. The optimal ensemble for $C_{\chi }(\Phi )$
is $\left\{ \pi ^{0}(dg),\,V_{g}\rho _{0}V_{g}^{\ast }\right\} $ where $\rho
_{0}$ is a minimizer for $S(\Phi (\rho ))$. If the group $G$ is not finite,
then this is a generalized ensemble in the sense of the next Section.

Assume moreover that additivity (\ref{maddi}) of the minimal output entropy
holds, then
\begin{eqnarray*}
&&n\left[ \max_{\rho }S\left( \Phi \left[ \rho \right] \right) -\check{S}(\Phi
)\right] = nC_{\chi }(\Phi )\leq C_{\chi }(\Phi ^{\otimes n}) \\
&&\leq \max_{\rho ^{(n)}}S\left( \Phi ^{\otimes n}\left[ \rho ^{(n)}\right]
\right) -\check{S}\left( \Phi ^{\otimes n}\right) \\
&&=n\left[ \max_{\rho }S\left( \Phi \left[ \rho \right] \right) -\check{S}%
(\Phi )\right] ,
\end{eqnarray*}%
where the first equality follows from (\ref{equ}), the first inequality --
from the definition of $C_{\chi },$ the second inequality -- from (\ref{ine1}%
) applied to $\Phi ^{\otimes n}$, the second equality -- from the equality (\ref{maxs})
(for irreducibly covariant channels it is just a consequence of (\ref
{unit}), and from the assumption (\ref{maddi}). Thus $C_{\chi
}(\Phi ^{\otimes n})=nC_{\chi }(\Phi )=n\left[ S\left( \Phi \left[ I/d\right]
\right) -\check{S}(\Phi )\right] $ and
\begin{equation*}
C(\Phi )=C_{\chi }(\Phi )=S\left( \Phi \lbrack I/d]\right) -\check{S}\left(
\Phi \right).
\end{equation*}

\section{Infinite-dimensional case}

Let $\mathcal{H}$ be a separable complex Hilbert space, $\mathfrak{L}(
\mathcal{H} ) $ the algebra of all bounded operators in $\mathcal{H}$, $
\mathfrak{T }(\mathcal{H})$ the space of trace-class operators, and $
\mathfrak{S}(\mathcal{H})$ the convex set of density operators in $
\mathcal{H}$. \textit{Quantum channel} is a linear completely positive
trace-preserving map $\Phi$ in $\mathfrak{T}(\mathcal{H}).$

\textit{Generalized ensemble} is a pair $\{\pi(dx), \rho (x)\}$ where $\pi$
is a probability measure on a standard Borel space $\mathcal{X}$ and $%
x\rightarrow \rho (x)$ is a measurable map from $\mathcal{X}$ to $\mathfrak{S%
}(\mathcal{H})$. The \textit{average state} of the generalized ensemble $\pi
$ is defined as the barycenter of the probability measure
\begin{equation*}
\bar{\rho}_{\pi }=\int\limits_{\mathcal{X}}\rho (x)\,\pi (d x ).
\end{equation*}%
The conventional ensembles correspond to finitely supported measures.

In the infinite-dimensional case one usually has to consider the input
constraints to avoid infinite values of the capacities. Let $H$ be a
positive selfadjoint operator in $\mathcal{H}$, which usually represents
energy of the input. We consider the input states with constrained energy: $%
\mathrm{Tr}\rho H\leq E,$ where $E$ is a fixed positive constant. Since the
operator $H$ can be unbounded, care should be taken in defining the trace;
we put $\mathrm{Tr}\rho H=\int_{0}^{\infty }\lambda \,dm_{\rho }(\lambda ),$
where $m_{\rho }(\lambda )=\mathrm{Tr}\rho E(\lambda ),$ and $E(\lambda )$
is the spectral function of the selfadjoint operator $H.$ Then the \textit{%
constrained $\chi -$ capacity} is given by the following generalization of
the expression (\ref{chicap}):
\begin{equation}
C_{\chi }(\Phi ,H,E)=\sup_{\pi :\mathrm{Tr}\bar{\rho}_{\pi }H\leq E}\chi
(\pi ),  \label{chi}
\end{equation}%
where%
\begin{equation}
\chi (\pi )=S(\Phi [ \bar{\rho}_{\pi }])-\int\limits_{\mathcal{X}} S(\Phi [
\rho (x)])\pi (dx ) .  \label{chipi}
\end{equation}%
To ensure that this expression is defined correctly, certain additional
conditions should  be imposed upon the channel $\Phi $ and the constraint
operator $H$ (see \cite{HOSHI}, \cite{h}, Sec. 11.5), which are always
fulfilled in the Gaussian case we consider below.

Denote $H^{(n)}=H\otimes I\dots \otimes I+\dots +I\otimes \dots \otimes
I\otimes H,$ then the \textit{constrained classical capacity} is given by
the expression
\begin{equation}
C(\Phi ,H,E)=\lim_{n\rightarrow \infty }\frac{1}{n}C_{\chi }(\Phi ^{\otimes
n},H^{(n)},nE)  \label{ccc}
\end{equation}
generalizing (\ref{hsw}) \cite{h}.

Consider the following constrained set of states
\begin{equation*}
\mathcal{E}_{E}=\left\{ \rho :\mathrm{Tr}\rho H\leq E\right\} ,
\end{equation*}%
We have an obvious estimate
\begin{equation}
C_{\chi }(\Phi ,H,E)\leq \sup_{\rho \in \mathcal{E}_{E}}S\left( \Phi \left[
\rho \right] \right) -\inf_{\rho }S(\Phi \left[ \rho \right] ).  \label{ine}
\end{equation}

\begin{proposition}
\label{prop1} Consider the following assumptions:

\begin{enumerate}
\item $\sup_{\rho \in \mathcal{E}_E}S\left( \Phi \left[ \rho \right] \right)
$ is attained on a state ${\rho}_{E}^0;$

\item $\inf_{\rho }S(\Phi (\rho ))$ is attained on a state ${\rho}_0 $;

\item $\Phi $ is a covariant channel in the sense (\ref{cova}), and there
exists a Borel probability measure ${\pi }_{E}^{0}$ on $G$ such that
\begin{equation*}
{\rho }_{E}^{0}=\int_{G}V_{g}{\rho }_{0}V_{g}^{\ast }\,{\pi }_{E}^{0}(dg).
\end{equation*}

\item the minimal output entropy of the channel $\Phi $ is additive in the
sense (\ref{maddi}),
\end{enumerate}

Then under the conditions 1-3
\begin{eqnarray}
&&C_{\chi }(\Phi ,H,E) \nonumber \\
&&=\sup_{\rho \in \mathcal{E}_{E}}S\left( \Phi \left[ \rho
\right] \right) -\inf_{\rho }S(\Phi \left[ \rho \right] )\nonumber \\
&&=S\left( \Phi \left[
{\rho }_{E}^{0}\right] \right) -S(\Phi \left[ {\rho }_{0}\right] ),
\label{C1}
\end{eqnarray}
and the optimal ensemble for $C_{\chi }$ consists of the states $V_{g}{\rho }%
_{0}V_{g}^{\ast }$ with the probability distribution ${\pi }_{E}^{0}(dg).$

If, in addition, the condition 4 holds, then
\begin{equation*}
C_{\chi }(\Phi ^{\otimes n},H^{(n)},nE)=nC_{\chi }(\Phi ,H,E)
\end{equation*}%
and
\begin{equation*}
C(\Phi ,H,E)=C_{\chi }(\Phi ,H,E)=S\left( \Phi \left[ {\rho }_{E}^{0}\right]
\right) -S(\Phi \left[ {\rho }_{0}\right] ).
\end{equation*}
\end{proposition}

\begin{proof} To prove the first statement it is sufficient to substitute
the ensemble $\left\{ {\pi }_{E}^{0}(dg),V_{g}\rho _{0}V_{g}^{\ast }\right\}
$ into the expression (\ref{chipi}). For covariant channels the integral
term is equal to $S(\Phi \lbrack \rho _{0}])$, thus we obtain that the
right-hand side of (\ref{ine}) is also a lower estimate for $C_{\chi }(\Phi
,H,E)$.

To prove the second statement we use lemma 11.20 of \cite{h}

\begin{lemma}
\label{L6}
\begin{equation*}
\sup_{\rho ^{(n)}:\mathop{\rm Tr}\nolimits\rho ^{(n)}H^{(n)}\leq nE}S(\Phi
^{\otimes n}[\rho ^{(n)}])=n\sup_{\rho :\mathop{\rm Tr}\nolimits\rho H\leq
E}S(\Phi \lbrack \rho ]).
\end{equation*}
\end{lemma}

\begin{proof}
We give the proof here for completeness. We first show that
\begin{equation}
\sup_{\rho ^{(n)}:\mathop{\rm Tr}\nolimits\rho ^{(n)}H^{(n)}\leq nE}S(\Phi
^{\otimes n}[\rho ^{(n)}])\leq n\sup_{\rho :\mathop{\rm Tr}\nolimits\rho
H\leq E}S(\Phi \lbrack \rho ]).  \label{hphine}
\end{equation}

Indeed, denoting by $\rho _{j}$ the partial state of $\rho ^{(n)}$ in the $%
j- $th tensor factor of ${\mathcal{H}}_{A}^{\otimes n}$ and letting $\bar{%
\rho}=\frac{1}{n}\sum_{j=1}^{n}\rho _{j},$ we have
\begin{equation*}
S(\Phi ^{\otimes n}[\rho ^{(n)}])\leq \sum_{j=1}^{n}S(\Phi \lbrack \rho
_{j}])\leq nS(\Phi \lbrack \bar{\rho}]),
\end{equation*}%
where in the first inequality we used subadditivity of the quantum entropy,
while in the second -- its concavity. Moreover,
$$\mathop{\rm Tr}\nolimits%
\bar{\rho}H=\frac{1}{n}\mathop{\rm Tr}\nolimits\rho ^{(n)}H^{(n)}\leq E,$$
hence (\ref{hphine}) follows. In the opposite direction, take $\rho
^{(n)}=\rho _{E}^{0\otimes n}$ and use the additivity of the entropy for product
states.
\end{proof}

Then we have similarly to the finite-dimensional case
\begin{widetext}
\begin{eqnarray*}
n\left[ S\left( \Phi \left[ \rho _{E}^{0}\right] \right) -S(\Phi \left[
\rho _{0}\right] )\right] &=& nC_{\chi }(\Phi ,H,E)\leq C_{\chi }(\Phi
^{\otimes n},H^{(n)},nE) \leq \max_{\rho ^{(n)}:\mathrm{Tr}\rho ^{(n)}H^{(n)}\leq
nE}S\left( \Phi ^{\otimes n}\left[ \rho ^{(n)}\right] \right) -\min_{\rho
^{(n)}}S\left( \Phi ^{\otimes n}\left[ \rho ^{(n)}\right] \right) \\
&=&n\left[ \max_{\rho :\mathrm{Tr}\rho H\leq E}S\left( \Phi \left[ \rho
\right] \right) -\min_{\rho }S\left( \Phi \left[ \rho \right] \right) \right]
=n\left[ S\left( \Phi \left[ \rho _{E}^{0}\right] \right) -S(\Phi \left[
\rho _{0}\right] )\right] ,
\end{eqnarray*}\end{widetext}
where the first equality follows from (\ref{C1}), the first inequality from
the definition of $C_{\chi },$ the second inequality from (\ref{ine})
applied to $\Phi ^{\otimes n}$, the second equality from lemma \ref{L6} and
the condition 4. Summarizing,
\begin{eqnarray*}
&&C_{\chi }(\Phi ^{\otimes n},H^{(n)},nE)\\
&&=nC_{\chi }(\Phi ,H,E)\\
&&=n\left[
S\left( \Phi \left[ \rho _{E}^{0}\right] \right) -S(\Phi \left[ \rho _{0}
\right] )\right] ,
\end{eqnarray*}
hence the second statement follows.\end{proof}

\section{The case of Bosonic Gaussian channels}

In the papers \cite{ghg}, \cite{ANDREA}, \cite{ggch} a solution of the
long-standing quantum Gaussian optimizers conjecture was given for gauge-covariant or
contravariant Bosonic Gaussian channels. In particular, the constrained
classical capacity was computed under the assumption that the constraint
operator is gauge-invariant with respect to the same complex structure as
the channel. Basing on observations of previous section and using the fact
that a general Bosonic Gaussian channel is irreducibly covariant under the
group of displacements (the Weyl group), we can relax this assumption.

In this section we systematically use notations and some results from
the book \cite{h} where further references are given (see also Appendix).
Let $\mathcal{H}$ be the space of an irreducible representation $
z\rightarrow W(z);\,z\in Z, $ of the Canonical Commutation Relations,
where $Z$ is a finite-dimensional symplectic space $(\mathbb{R}^{2s},\Delta )
$ with
\begin{equation}
\Delta =\left[
\begin{array}{ccccc}
0 & 1 &  &  &  \\
-1 & 0 &  &  &  \\
&  & \ddots &  &  \\
&  &  & 0 & 1 \\
&  &  & -1 & 0%
\end{array}%
\right] \equiv \mathrm{diag}\left[
\begin{array}{cc}
0 & 1 \\
-1 & 0%
\end{array}%
\right] .  \label{delta}
\end{equation}%
Here $W(z)=\exp i\,Rz$ are the unitary Weyl operators, where
\begin{equation*}
Rz=\sum_{j=1}^{s}(x_{j}q_{j}+y_{j}p_{j}),
\end{equation*}%
and
\begin{equation*}
R=[q_{1}\;p_{1}\;\dots \;q_{s}\;p_{s}]
\end{equation*}%
is the row vector of the canonical observables of the quantized system,
while $z=[x_{1}\;y_{1}\;\dots \;x_{s}\;y_{s}]^{t}$ is the column vector of
the real parameters.

A centered Gaussian state $\rho $ on $\mathfrak{L}(\mathcal{H})$ is
determined by its covariance matrix $\alpha =\mathrm{Re\,}\mathrm{Tr\,}%
R^{t}SR$ which is a real symmetric $s\times s$-matrix satisfying
\begin{equation*}
\alpha \geq \pm \frac{i}{2}\Delta
\end{equation*}%
The entropy of $\rho $ is equal to
\begin{equation}
S(\rho )=\frac{1}{2}\mathrm{Sp}\ g\left( \mathrm{abs}(\Delta ^{-1}\alpha )-%
\frac{I}{2}\right) ,  \label{abs}
\end{equation}%
where $\mathrm{Sp}$ is used to denote trace of a matrix as distinct from the
trace $\mathrm{Tr}$ of operators in the underlying Hilbert space \cite{hir}.
The operator $A=\Delta ^{-1}\alpha $ has the eigenvalues $\pm i\alpha _{j}$
(where $\alpha _{j}$ are real), hence its matrix is diagonalizable (in the
complex domain). For any diagonalizable matrix $M=U\mathrm{\ diag}%
(m_{j})U^{-1}$, we denote $\mathrm{abs}(M)=U\mathrm{diag}(|m_{j}|)U^{-1}$.

Let $H=\sum_{j,k=1}^{s}\epsilon _{jk}R_{j}R_{k},$ where $\epsilon =\left[
\epsilon _{jk}\right] $ is a symmetric positive definite matrix, be a
quadratic energy operator. Notice that it always has an associated complex
structure $J_H$ satisfying $[\epsilon\Delta ,J_H]=0$. This is the orthogonal
operator from the polar decomposition of the operator $-\epsilon\Delta$, see
sec. 12.2.3 \cite{h} for detail.

A centered Gaussian channel $\Phi $ is defined by the relation
\begin{equation}
\Phi ^{\ast }[W(z)]=W(Kz)\exp \left(-\frac{1}{2}z^{t}\mu z\right) .
\label{bosgaus}
\end{equation}%
where $(K,\,\mu )$ are the matrix parameters satisfying
\begin{equation*}
\mu \geq \pm \frac{i}{2}\left( \Delta -K^{t}\Delta K\right) .
\end{equation*}%
The action of the channel on the centered Gaussian state with a covariance
matrix $\alpha $ is described by the equation%
\begin{equation*}
\alpha \rightarrow K^{t}\alpha K+\mu .
\end{equation*}%
A Bosonic Gaussian channel is irreducibly covariant with respect to the
representation $z\rightarrow W(z) $ in the sense
\begin{equation*}
\Phi \lbrack W(z)\rho W(z)^{\ast }]=W(K^{s}z)\Phi \lbrack \rho
]W(K^{s}z)^{\ast },\quad z\in Z,
\end{equation*}%
where $K^{s}=\Delta ^{-1}K^{t}\Delta $, see e.g. sec. 12.4.2 in \cite{h}.

The conditions 1 of proposition \ref{prop1} follows from the argument of
sec. 12.5 \cite{h}, moreover $\rho _{E}^{0}$ is a centered Gaussian state
with a covariance matrix
\begin{equation*}
\alpha _{E}^{0}=\arg \max_{\alpha :\,\mathrm{Sp}\alpha \epsilon \leq E}%
\mathrm{Sp}\ g\left( \mathrm{abs}(\Delta ^{-1}\left[ K^{t}\alpha K+\mu %
\right] )-\frac{I}{2}\right)
\end{equation*}

Assuming that the channel $\Phi $ is gauge-covariant or contravariant with respect to a
complex structure $J$ in $Z,$ the conditions 2 and 4 follow from the results
of the paper \cite{ghg} concerning the minimal output entropy. Moreover, $%
\rho _{0}$ can be taken as the vacuum state related to the complex structure
$J$. It is shown in sec. 12.3.2 of \cite{h} that the vacuum state related to
the complex structure $J$ is the pure centered Gaussian state with the
covariance matrix $\frac{1}{2}\Delta J.$

The condition 3 is fulfilled provided
\begin{equation}
{\alpha } _{E}^{0}\geq \frac{1}{2}\Delta J.  \label{thr}
\end{equation}%
In this case%
\begin{equation*}
\rho _{E}^{0}=\int_{Z}W(z)\rho _{0}W(z)^{\ast }\,{\pi }_{E}^{0}(dz),
\end{equation*}%
where ${\pi }_{E}^{0}(dz)$ is the centered Gaussian distribution on $Z$ with the
covariance matrix ${\alpha } _{E}^{0}-\frac{1}{2}\Delta J.$ One can check this by
comparing the quantum characteristic functions of both sides. The optimizing
ensemble consists thus of the $J-$coherent states $W(z)\rho _{0}W(z)^{\ast }$
with the probability distribution ${\pi }_{E}^{0}(dz).$

The constrained classical capacity of the channel $\Phi $ is equal to
\begin{eqnarray}
&&C(\Phi ;H,E)=C_{\chi }(\Phi ;H,E)  \label{maxnu} \\
&&\frac{1}{2}\max_{\alpha :\,\mathrm{Sp}\alpha
\epsilon \leq E}\mathrm{Sp}\ g\left( \mathrm{abs}(\Delta ^{-1}\left[
K^{t}\alpha K+\mu \right] )-\frac{I}{2}\right)  \notag \\
&&-\frac{1}{2}\mathrm{Sp}\ g\left( \mathrm{abs}(\Delta ^{-1}\left[ \frac{1}{2}%
K^{t}\Delta JK+\mu \right] )-\frac{I}{2}\right) .  \notag
\end{eqnarray}

Gauge-covariance of the channel $\Phi$ with respect to a complex structure $%
J $ is equivalent to the conditions
\begin{equation}
\lbrack K,J]=0,\quad \lbrack \Delta ^{-1}\,\mu ,J]=0.  \label{cov}
\end{equation}%
Given a symmetric $\mu > 0,$ one can always find a complex structure $J,$
satisfying $[\Delta ^{-1}\,\mu ,J]=0;$ it is just the orthogonal operator
from the polar decomposition of the operator $\Delta ^{-1}\,\mu $ in the
Euclidean space $(Z,\mu )$, cf. (\ref{kspd}). Then the first equation becomes a restriction
for admissible $K.$ For gauge-contravariant channels it is replaced by $%
\{K,J\}=0$.

In the paper \cite{schaefer} the \textit{Gaussian capacities} obtained by
optimization over Gaussian inputs where computed for a generic
non-degenerate single-mode channel when the input signal energy is above
certain threshold. Our observations imply in particular that these Gaussian
capacities are in fact equal to the full classical capacities, and the
inequality (\ref{thr}) appears as the multimode generalization of the
threshold condition in \cite{schaefer}.

Let us confirm this by calculation of the example of squeezed noise channel.
The channel is described by the parameters
\begin{equation*}
K=k\left[
\begin{array}{cc}
1 & 0 \\
0 & 1%
\end{array}%
\right] ,\quad \mu =\left[
\begin{array}{cc}
\mu _{1} & 0 \\
0 & \mu _{2}%
\end{array}%
\right] ;\quad \mu _{1}\mu _{2}\geq \frac{1}{4}\left\vert k^{2}-1\right\vert
^{2}.
\end{equation*}%
This describes attenuation $(0<k<1)$, amplification $(1<k)$ and additive
noise $(k=1)$ channels, with the background squeezed noise. Take the energy
operator $H=q^{2}+p^{2}$ with $\epsilon=\left[
\begin{array}{cc}
1 & 0 \\
0 & 1
\end{array}\right] $ and the corresponding complex structure $J_{H}$ $=%
\left[
\begin{array}{cc}
0 & -1 \\
1 & 0%
\end{array}%
\right] .$

The complex structure of the channel satisfying (\ref{cov}) is given by
\begin{equation*}
J=\left[
\begin{array}{cc}
0 & -\sqrt{\mu _{2}/\mu _{1}} \\
\sqrt{\mu _{1}/\mu _{2}} & 0%
\end{array}%
\right] ,
\end{equation*}%
which does not commute with $J_{H}$ unless $\mu _{1}=\mu _{2}.$ The
covariance matrix of the squeezed vacuum is
\begin{equation*}
\frac{1}{2}\Delta J=\frac{1}{2}\left[
\begin{array}{cc}
\sqrt{\mu _{1}/\mu _{2}} & 0 \\
0 & \sqrt{\mu _{2}/\mu _{1}}%
\end{array}%
\right] .
\end{equation*}%
The eigenvalues of the matrix
\begin{eqnarray*}
&&\Delta ^{-1}\left[ \frac{1}{2}K^{t}\Delta JK+\mu \right]\\
&& =\left[
\begin{array}{cc}
0 & \mu _{2}+\frac{k^{2}}{2}\sqrt{\mu _{2}/\mu _{1}} \\
-\left( \mu _{1}+\frac{k^{2}}{2}\sqrt{\mu _{1}/\mu _{2}}\right) & 0
\end{array}
\right]
\end{eqnarray*}%
are equal to $\pm i\left( \sqrt{\mu _{1}\mu _{2}}+k^{2}/2\right) ,$ hence
the second term in (\ref{maxnu}) is $g\left( \sqrt{\mu _{1}\mu _{2}}+\left(
k^{2}-1\right) /2\right) .$

To compute the first term, we can restrict to diagonal covariance matrices
\begin{equation*}
\alpha =\left[
\begin{array}{cc}
\alpha _{1} & 0 \\
0 & \alpha _{2}%
\end{array}%
\right] ,\quad \alpha _{1}+\alpha _{2}\leq E,\quad \alpha _{1}\alpha
_{2}\geq \frac{1}{4}.
\end{equation*}%
The matrix%
\begin{equation*}
\Delta ^{-1}\left[ K^{t}\alpha K+\mu \right] =\left[
\begin{array}{cc}
0 & \mu _{2}+k^{2}\alpha _{2} \\
-\left( \mu _{1}+k^{2}\alpha _{1}\right) & 0%
\end{array}%
\right]
\end{equation*}%
has the eigenvalues $\pm i\sqrt{\left( \mu _{1}+k^{2}\alpha _{1}\right)
\left( \mu _{2}+k^{2}\alpha _{2}\right) },$ so that the maximized expression
is $$g\left( \sqrt{\left( \mu _{1}+k^{2}\alpha _{1}\right) \left( \mu
_{2}+k^{2}\alpha _{2}\right) }-1/2\right) .$$ Since $g(x)$ is increasing, we
have to maximize $\left( \mu _{1}+k^{2}\alpha _{1}\right) \left( \mu
_{2}+k^{2}\alpha _{2}\right) $ under the constraints $\alpha _{1}+\alpha
_{2}\leq E,\quad \alpha _{1}\alpha _{2}\geq \frac{1}{4}.$ The first
constraint gives the values
\begin{equation*}
\alpha ^{0}_{1}=E/2+\left( \mu _{2}-\mu _{1}\right) /2k^{2},\quad \alpha
^{0}_{2}=E/2-\left( \mu _{2}-\mu _{1}\right) /2k^{2}
\end{equation*}%
corresponding to the maximal value of the first term
\begin{equation*}
g\left( \frac{1}{2}\left( k^{2}E+\left( \mu _{1}+\mu _{2}\right) -1\right)
\right) .
\end{equation*}%
The second constraint will be automatically fulfilled provided we impose the
condition (\ref{thr}) which amounts to $\alpha _{1}^{0}\geq \frac{1}{2}\sqrt{%
\mu _{1}/\mu _{2}},\,\alpha _{2}^{0}\geq \frac{1}{2}\sqrt{\mu _{2}/\mu _{1}}%
, $ or, introducing the squeezing parameter $\eta =\sqrt{\mu _{2}/\mu
_{1}} $%
\begin{equation*}
E\geq \frac{1}{2}\left[ \eta +\eta ^{-1}+\left\vert \eta -\eta
^{-1}\right\vert \left( 1+\frac{2}{k^{2}}\sqrt{\mu _{1}\mu _{2}}\right) %
\right] ,
\end{equation*}
where the term $\frac{1}{2}(\eta +\eta ^{-1})$ corresponds to the energy of the squeezed vacuum
$\mathrm{Sp}\left(\frac{1}{2}\Delta J\right)\epsilon$.
Under this condition
\begin{widetext}
\begin{equation*}
C(\Phi ;H,E)=C_{\chi }(\Phi ;H,E)=g\left( \frac{1}{2}\left( k^{2}E+\left(
\mu _{1}+\mu _{2}\right) -1\right) \right) -g\left( \sqrt{\mu _{1}\mu _{2}}%
+\left( k^{2}-1\right) /2\right) .
\end{equation*}
\end{widetext}
These values up to notations coincide with those computed in \cite{schaefer}
, Cor. 2.

\section{Implications of gauge-covariance for nondegenerate Gaussian channels}

The solution of the Gaussian optimizers conjecture given in \cite
{ghg} applies to a more general situation where the complex structures $J_{A}, J_{B}$
at the input and the output of the channel need not coincide. In that case the gauge
covariance condition (\ref{cov}) take the form (see \cite{h}, Eq. (12.147))
\begin{equation}
KJ_{B}-J_{A}K=0,\quad \lbrack \Delta ^{-1}\mu ,J_{B}]=0.  \label{kj}
\end{equation}%
Let us investigate how restrictive is this condition, assuming that both $%
\mu $ and $K$ are nondegenerate matrices. As already mentioned, there exist $%
J_{B}$ satisfying the second equation, therefore we ask for existence of the
complex structure $J_{A}$ satisfying the first equation. Since $K$ is
nondegenerate, this is equivalent to the question: when does the operator  $%
J=KJ_{B}K^{-1}$ define a complex structure? Notice that always $J^{2}=-I,$
therefore a necessary and sufficient condition for that is (\ref{comstr}),
which in our case transforms into
\begin{equation*}
\Delta KJ_{B}K^{-1}=-\left( K^{t}\right) ^{-1}J_{B}^{t}K^{t}\Delta \geq 0.
\end{equation*}%
Taking into account nondegeneracy of $K,$ this is equivalent to%
\begin{equation*}
K^{t}\Delta KJ_{B}=-J_{B}^{t}K^{t}\Delta K\geq 0
\end{equation*}%
or%
\begin{equation}
\left( K^{t}\Delta K\Delta ^{-1}\right) \Delta J_{B}=-J_{B}^{t}\Delta \left(
K^{t}\Delta K\Delta ^{-1}\right) ^{t}\geq 0.  \label{commut}
\end{equation}%
Notice that $\Delta J_{B}=-J_{B}^{t}\Delta \geq 0$ since $J_{B}$ is  a
complex structure satisfying (\ref{comstr}).

Let us compute the operator $K^{t}\Delta K\Delta ^{-1}.$ It will be
convenient to rearrange the coordinates so that $z=[x_{1}\;\dots
\,x_{s}\,y_{1}\;\dots \;y_{s}]^{t},$ then
\begin{equation*}
\Delta =\left[
\begin{array}{cc}
0 & I \\
-I & 0%
\end{array}%
\right] ,\quad K=\left[
\begin{array}{cc}
A & B \\
C & D%
\end{array}%
\right] ,
\end{equation*}%
and%
\begin{equation}
K^{t}\Delta K\Delta ^{-1}=\left[
\begin{array}{cc}
A^{t}D-C^{t}B & C^{t}A-A^{t}C \\
B^{t}D-D^{t}B & D^{t}A-B^{t}C%
\end{array}%
\right] .  \label{op}
\end{equation}%
In the case of single mode, $A,B,C,D$ are real numbers, hence
\begin{equation*}
K^{t}\Delta K\Delta ^{-1}=\det K\left[
\begin{array}{cc}
1 & 0 \\
0 & 1%
\end{array}%
\right]
\end{equation*}%
and the condition (\ref{commut}) is fulfilled provided $\det K>0$.
We thus obtain a result which is implicit in Sec. 12.6.1 of \cite{h}:
all the nondegenerate single-mode Gaussian channels with
positive $\det K$ are gauge-covariant. In the case $\det K>0$ we obtain
gauge-contravariant channels with $KJ_{B}+J_{A}K=0$ instead of the first
equation in (\ref{kj}).

In the case of many modes, $s\geq 2$, by asking the operator (\ref{op}) to be a multiple
of identity, we obtain a sufficient condition for (\ref{commut}), and hence
for the channel to be gauge-covariant. On the other case, for many modes it
is easy to give an example of a nondegenerate Gaussian channel which is not
gauge-covariant  for any choice of complex structures $J_{A},J_{B}.$ Consider
the canonical complex structure $J_{B}=\Delta ^{-1}$. Then the condition (\ref{commut}) is equivalent to the requirement that the matrix (\ref{op})
is symmetric and nonnegative definite. Take a positive diagonal matrix $\mu $ satisfying
the second condition in (\ref{kj}) and $K=\left[
\begin{array}{cc}
I & 0 \\
0 & D%
\end{array}
\right] .$ The equality in  (\ref{commut}) is not fulfilled  unless $D^{t}=D.
$ In particular, this is the case for ``not-so-normal mode
decomposition'' of $K$ \cite{wolf} with non-diagonal Jordan
blocks.

\section{Appendix. Complex structures in a symplectic space}

Denote $\Delta (z,z^{\prime })=z^{t}\Delta z^{\prime }$ the symplectic form
in $Z.$ A basis $\left\{ e_{j},h_{j};j=1,\dots ,s\right\} $ in $Z$ is called
\textit{symplectic} if
\begin{equation}
\Delta (e_{j},h_{k})=\delta _{jk},\quad\Delta (e_{j},e_{k})=\Delta
(h_{j},h_{k})=0  \label{syba}
\end{equation}
for all $j,k=1,\dots ,s$. The transition matrix $T$ from the initial symplectic basis in $Z$ to the
new symplectic basis is a matrix of \textit{symplectic transformation} in $%
(Z,\Delta ),$ which is characterized by the property
\begin{equation*}
\Delta (Tz,Tz^{\prime })=\Delta (z,z^{\prime });\quad z,z^{\prime }\in Z.
\end{equation*}

Operator $J$ in $(Z,\Delta )$ is called \textit{operator of complex structure%
} if
\begin{equation}
J^{2}=-I,  \label{j2e}
\end{equation}%
where $I$ is the identity operator in $Z$, and it is $\Delta -$positive in
the sense that the bilinear form
\begin{equation}
j(z,z^{\prime })=\Delta (z,Jz^{\prime })  \label{inpr}
\end{equation}%
is an inner product in $Z$. Note that $\Delta -$positivity is equivalent to
the conditions
\begin{equation}
\Delta J=-J^{t}\Delta ,\quad \Delta J\geq 0.  \label{comstr}
\end{equation}

Operator $J$ defines the structure of complex unitary space in $Z$ (of
dimensionality $s$) in which $iz=Jz$ and the inner product is
\begin{equation*}
j(z,z^{\prime })+i\Delta (z,z^{\prime })=\Delta (z,Jz^{\prime })+i\Delta
(z,z^{\prime }).
\end{equation*}

In what follows we will consider bilinear forms $\alpha ,\mu $ in the space $%
Z=\mathbb{R}^{2s}$, and the matrices of such forms will be denoted by the
same letters, e.g. $\alpha (z,z^{\prime })=z^{t}\alpha z^{\prime },$ etc.

\begin{lemma}
\label{symbas} Let $\alpha \left( z,z^{\prime }\right) =z^{t}\alpha
z^{\prime}$ be an inner product in the symplectic space $(Z,\Delta )$. Then
there is a symplectic basis $\left\{ e_{j},h_{j};j=1,\dots ,s\right\} $ in $%
Z $ such that the form $\alpha $ is diagonal with the matrix
\begin{equation}  \label{salpha}
\tilde{\alpha}=\mathrm{diag}\left[
\begin{array}{cc}
\alpha _{j} & 0 \\
0 & \alpha _{j}%
\end{array}
\right] ,
\end{equation}
where $\alpha _{j}>0.$
\end{lemma}

Consider the operator $A=\Delta ^{-1}\alpha $ satisfying
\begin{equation*}
\alpha (z,z^{\prime })=\Delta (z,Az^{\prime }).
\end{equation*}%
The operator $A$ is skew-symmetric in the Euclidean space $(Z,\alpha
):A^{\ast }=-A.$ According to a theorem from linear algebra, there is an
orthogonal basis $\left\{ e_{j},h_{j}\right\} $ in $(Z,\alpha )$ and
positive numbers $\left\{ \alpha _{j}\right\} $ such that
\begin{equation*}
Ae_{j}=\alpha _{j}h_{j};\quad Ah_{j}=-\alpha _{j}e_{j}.
\end{equation*}%
Choosing the normalization $\alpha (e_{j},e_{j})=\alpha (h_{j},h_{j})=\alpha
_{j}$ gives the symplectic basis in $(Z,\Delta )$ with the required
properties.

For arbitrary inner product $\alpha $ in $Z$ there is at least one operator
of complex structure $J$, commuting with the operator $A=\Delta ^{-1}\alpha
, $ namely, the orthogonal operator $J$ from the polar decomposition
\begin{equation}  \label{kspd}
A=\left\vert A \right\vert J=J\left\vert A \right\vert
\end{equation}
in the Euclidean space $(Z,\alpha ).$ Applying Lemma \ref{symbas}, we obtain
that there is a symplectic basis $\left\{ e_{j},h_{j};j=1,\dots ,s\right\} $
in which the form $\alpha $ is diagonal with the matrix (\ref{salpha}) while
$J$ has the matrix
\begin{equation}
\tilde{J}=\mathrm{diag}\left[
\begin{array}{cc}
0 & -1 \\
1 & 0%
\end{array}
\right] ,  \label{jtilde}
\end{equation}
so that
\begin{equation*}
J e_j=h_j,\quad J h_j=-e_j.
\end{equation*}

With every complex structure we can associate the cyclic one-parameter group
$\left\{ \mathrm{e}^{\varphi J};\varphi \in \lbrack 0,2\pi ]\right\} $ of
symplectic transformations which we call the \textit{gauge group}. The gauge
group in $Z$ induces the unitary group of the \textit{gauge transformations}
in $\mathcal{H}$ by the formula
\begin{equation}
W(\mathrm{e}^{\varphi J}z)=\mathrm{e}^{-i\varphi G}W(z)\mathrm{e}^{i\varphi
G},  \label{gaugetr}
\end{equation}%
where $G=\sum_{j=1}^{s}\tilde{a}_{j}^{\dagger }\tilde{a}_{j}$ is the
\textit{total number operator} in $\mathcal{H}$.

An operator $X$ in $\mathcal{H}$ is called \textit{gauge invariant} if
\begin{equation*}
\mathrm{e}^{-i\varphi G}X\mathrm{e}^{i\varphi G}=X
\end{equation*}%
for all $\varphi \in \lbrack 0,2\pi ].$ By using (\ref{gaugetr}) and (\ref%
{comstr}) we find that a quadratic operator $X=R\epsilon R^{T}$, where $%
\epsilon $ is a symmetric positive matrix, is gauge invariant if
\begin{equation*}
\lbrack J,\epsilon \Delta ]=0,
\end{equation*}%
i.e. $J$ is the operator of complex structure from the polar decomposition
of the skew-symmetric operator $-\epsilon \Delta $ in the Euclidean space $%
(Z, \epsilon )$ equipped with the inner product $\epsilon
(z,z^{\prime})=-z^t\Delta\epsilon \Delta z^{\prime}$. Such a complex
structure exists for every energy matrix $\epsilon $.

Consider a channel $\Phi $ in $\mathfrak{T}(\mathcal{H})$. Assume that in
the space $Z$ some operator of complex structure $J$ is fixed and let $G$ be
the operator generating the unitary group of gauge transformations in $%
\mathcal{H}$ according to the formula (\ref{gaugetr}). The channel is called
\textit{gauge covariant}, if
\begin{equation}
\Phi \lbrack \mathrm{e}^{i\varphi G}\rho \mathrm{e}^{-i\varphi G}]=\mathrm{e}%
^{i\varphi G}\Phi \lbrack \rho ]\mathrm{e}^{-i\varphi G}  \label{calinvch}
\end{equation}%
for all input states $\rho $ and all $\varphi \in \lbrack 0,2\pi ].$ For the
Gaussian channel with parameters $(K,\mu )$ this reduces to the conditions (\ref{cov}).
Thus, a natural choice of the complex structure in $Z$ is given by any $J$,
commuting with the operator $\Delta ^{-1}\mu $. Existence of such a complex
structure for nondegenerate matrix $\mu $ follows from the proof of lemma \ref{symbas}.

\section{Acknowledgments}

The author is grateful to M. E. Shirokov for comments and discussions. The
work was supported by the grant of Russian Scientific Foundation (project No
14-21-00162).

\end{document}